\documentclass[journal]{IEEEtran}
\usepackage{array}
\usepackage{amsmath}
\usepackage{amssymb}
\usepackage{upref}
\usepackage{amsfonts}
\usepackage{graphicx}
\usepackage{wasysym}

\usepackage{mathrsfs}

\newcommand{\deff}{\mbox{$\stackrel{\rm def}{=}$}}

\newcommand{\field}[1]{\mathbb{#1}}
\newcommand{\C}{\field{C}}
\newcommand{\F}{\field{F}}

\newcommand{\cF}{{\cal F}}

\newcommand{\cA}{{\cal A}}
\newcommand{\cB}{{\cal B}}
\newcommand{\cC}{{\cal C}}
\newcommand{\cG}{{\cal G}}

\newcommand{\cP}{{\cal P}}

\newcommand{\cX}{{\cal X}}

\newcommand{\sP}{\cP}
\newcommand{\sG}{\cG}

\newcommand{\Gr}{\smash{{\sG\kern-1.5pt}_q\kern-0.5pt(n,k)}}
\newcommand{\Grtwo}{\smash{{\sG\kern-1.5pt}_2\kern-0.5pt(n,k)}}
\newcommand{\Gkone}{\smash{{\sG\kern-1.5pt}_q\kern-0.5pt(n,k_1)}}
\newcommand{\Gktwo}{\smash{{\sG\kern-1.5pt}_q\kern-0.5pt(n,k_2)}}
\newcommand{\Ps}{\smash{{\sP\kern-2.0pt}_q\kern-0.5pt(n)}}

\newtheorem{theorem}{Theorem}

\newtheorem{lemma}{Lemma}
\newtheorem{remark}{Remark}
\newtheorem{cor}{Corollary}
\newtheorem{example}{Example}
\newtheorem{conjecture}{Conjecture}

\begin{document}

\bibliographystyle{IEEEtran}

\title{Error-Correcting Codes in Projective Spaces \\ via Rank-Metric Codes and Ferrers Diagrams}
\author{Tuvi
Etzion,~\IEEEmembership{Fellow,~IEEE} and Natalia Silberstein
%\thanks{The material in this paper was presented in part in the 2004 IEEE
%International Symposium on Information Theory, Chicago, USA.}
\thanks{T. Etzion is with the Department of Computer Science,
Technion --- Israel Institute of Technology, Haifa 32000, Israel.
(email: etzion@cs.technion.ac.il).}
\thanks{N. Silberstein is with the Department of Computer Science,
Technion --- Israel Institute of Technology, Haifa 32000, Israel.
(email: natalys@cs.technion.ac.il). This work is part of her Ph.D.
thesis performed at the Technion.}
\thanks{This work was supported in part by the Israel
Science Foundation (ISF), Jerusalem, Israel, under Grant 230/08.}
}

\maketitle

\begin{abstract}
Coding in the projective space has received recently a lot of
attention due to its application in network coding. Reduced row
echelon form of the linear subspaces and Ferrers diagram can play
a key role for solving coding problems in the projective space. In
this paper we propose a method to design error-correcting codes in
the projective space. We use a multilevel approach to design our
codes. First, we select a constant weight code. Each codeword
defines a skeleton of a basis for a subspace in reduced row
echelon form. This skeleton contains a Ferrers diagram on which we
design a rank-metric code. Each such rank-metric code is lifted to
a constant dimension code. The union of these codes is our final
constant dimension code. In particular the codes constructed
recently by Koetter and Kschischang are a subset of our codes.
The rank-metric codes used for this construction form a new class
of rank-metric codes. We present a
decoding algorithm to the constructed codes in the projective space.
The efficiency of the decoding depends on the efficiency of the decoding for
the constant weight codes and the rank-metric codes.
Finally, we use puncturing on our final constant dimension codes
to obtain large codes in the projective space which are not
constant dimension.
\end{abstract}


\section{Introduction}
\label{sec:introduction}

The {\it projective space} of order $n$ over the finite
field~\smash{$\F_q$}, denoted $\Ps$, is the set of all subspaces
of the vector space~\smash{$\F_q^n$}.  Given a nonnegative integer
$k \le n$, the set of all subspaces of \smash{$\F_q^n$} that have
dimension $k$ is known as a {\it Grassmannian}, and usually
denoted~by~$\Gr$. Thus, $ \Ps = \bigcup_{0 \le k \le n} \Gr $. It
turns out that~the natural~measure of distance in $\Ps$ is given
by
$$
d_S (U,\!V) \,\ \deff\ \dim U + \dim V -2 \dim\bigl( U\, {\cap}
V\bigr)
$$
for all $U,V \,{\in}\, \Ps$. It is well known (cf.\cite{AAK,KK})
that the function above is a metric;
%R on $\Ps$;
thus both $\Ps$ and $\Gr$~can~be regarded as metric spaces. Given
a metric~space, one can define codes. We say that $\C
\kern1pt{\subseteq}\kern1pt \Ps$ is an $(n,M,d)_q$ {\it code in
projective space} if $|\C| = M$ and $d_S (U,\!V) \ge d$ for all
$U,\!V \in \C$. If an $(n,M,d)_q$ code $\C$ is contained in $\Gr$
for some~$k$,~we say that $\C$ is an $(n,M,d,k)_q$ {\it constant
dimension code}. The $(n,M,d)_q$, respectively $(n,M,d,k)_q$,
codes in projective space are akin to the familiar codes in the
Hamming space, respectively (constant-weight) codes in the Johnson
space, where the Hamming distance serves as the metric.

Koetter and Kschischang~\cite{KK} showed that codes in $\Ps$ are
precisely what is needed for error-correction in random network
coding: an $(n,M,d)_q$ code can correct any $t$ packet errors (the
packet can be overwritten), which is equivalent to $t$ insertions
and $t$ deletions of dimensions in the transmitted subspace, and
any $\rho$ packet erasures introduced (adversarially) anywhere in
the network as long as $4t + 2\rho < d$ (see~\cite{SKK} for more
details). This is the motivation to explore error-correcting codes
in $\Ps$~\cite{EV,GaBo08,GaYa1,GaYa2,KoKu08,MGR08,Ska08,XiFu07}.
Koetter and Kschischang~\cite{KK} gave a Singleton like upper
bound on the size of such codes and a Reed-Solomon like code which
asymptotically attains this bound. Silva, Koetter, and
Kschischang~\cite{SKK} showed how these codes can be described in
terms of rank-metric codes~\cite{Gab85,Rot91}. The related
construction is our starting point in this paper. Our goal is to
generalize this construction in the sense that the codes of
Koetter and Kschischang will be sub-codes of our codes and all our
codes can be partitioned into sub-codes, each one of them is a
Koetter and Kschischang like code. In the process we describe some
tools that can be useful to handle other coding problems in $\Ps$.
We also define a new type of rank-metric codes and construct
optimal such codes. Our construction for constant dimension codes
and projective space codes uses a multilevel approach. This
approach requires a few concepts which will be described in the
following sections.

The rest of this paper is organized as follows. In
Section~\ref{sec:form} we define the reduced row echelon form of a
$k$-dimensional subspace and its Ferrers diagram. The reduced row
echelon form is a standard way to describe a linear subspace. The
Ferrers diagram is a standard way to describe a partition of a
given positive integer into positive integers. It appears that the
Ferrers diagrams can be used to partition the subspaces of $\Ps$
into equivalence classes~\cite{Knu71,vLWi92}. In
Section~\ref{sec:partial} we present rank-metric codes which will
be used for our multilevel construction. Our new method requires
rank-metric codes in which some of the entries are forced to be
zeroes due to constraints given by the Ferrers diagram. We first
present an upper bound on the size of such codes. We show how to
construct some rank-metric codes which attain this bound. In
Section~\ref{sec:eccGrass} we describe in details the multilevel
construction of the constant dimension codes. We start by
describing the connection of the rank-metric codes to constant
dimension codes. This connection was observed before
in~\cite{KK,SKK,GaYa1,GaYa2}. We proceed to describe the
multilevel construction. First, we select a binary constant weight
code ${\bf C}$. Each codeword of ${\bf C}$ defines a skeleton of a
basis for a subspace in reduced row echelon form. This skeleton
contains a Ferrers diagram on which we design a rank-metric code.
Each such rank-metric code is lifted to a constant dimension code.
The union of these codes is our final constant dimension code. We
discuss the parameters of these codes and also their decoding
algorithms. In Section~\ref{sec:eccProj} we generalize the
well-known concept of a punctured code for a code in the
projective space. Puncturing in the projective space is more
complicated than its counterpart in the Hamming space. The
punctured codes of our constant dimension codes have larger size
than the codes obtained by using the multilevel approach described
in Section~\ref{sec:eccGrass}. We discuss the parameters of the
punctured code and also its decoding algorithm. Finally, in
Section~\ref{sec:conclude} we summarize our results and present
several problems for further research.

\section{Reduced Echelon Form and Ferrers Diagram}
\label{sec:form}

In this section we give the definitions for two structures which
are useful in describing a subspace in $\Ps$. The reduced row
echelon form is a standard way to describe a linear subspace. The
Ferrers diagram is a standard way to describe a partition of a
given positive integer into positive integers.

A matrix is said to be in {\it row echelon form} if each nonzero
row has more leading zeroes than the previous row.

A $k \times n$ matrix with rank $k$ is in {\it reduced row echelon
form} if the following conditions are satisfied.
\begin{itemize}
\item The leading coefficient of a row is always to the right of
the leading coefficient of the previous row.

\item All leading coefficients are {\it ones}.

\item Every leading coefficient is the only nonzero entry in its
column.
\end{itemize}

A $k$-dimensional subspace $X$ of $\F_q^n$ can be represented by a
$k \times n$ {\it generator matrix} whose rows form a basis for
$X$. We usually represent a codeword of a projective space code by
such a matrix. There is exactly one such matrix in reduced row
echelon form and it will be denoted by $E(X)$.

\vspace{0.2cm}
\begin{example}
\label{exm:running} We consider the 3-dimensional subspace  $X$ of
$\F_2^7$ with the following eight elements.

\begin{footnotesize}
\begin{align*}
\begin{array}{cccccccc}
\text{1)} & (0 & 0 & 0 & 0 & 0 & 0 & 0) \\
\text{2)} & (1 & 0 & 1 & 1 & 0 & 0 & 0) \\
\text{3)} & (1 & 0 & 0 & 1 & 1 & 0 & 1) \\
\text{4)} & (1 & 0 & 1 & 0 & 0 & 1 & 1) \\
\text{5)} & (0 & 0 & 1 & 0 & 1 & 0 & 1) \\
\text{6)} & (0 & 0 & 0 & 1 & 0 & 1 & 1) \\
\text{7)} & (0 & 0 & 1 & 1 & 1 & 1 & 0) \\
\text{8)} & (1 & 0 & 0 & 0 & 1 & 1 & 0)
\end{array} .
\end{align*}
\end{footnotesize}
The basis of $X$ can be represented by a $3 \times 7$ matrix whose
rows form a basis for the subspace. There are 168 different
matrices for the 28 different basis. Many of these matrices are in
row echelon form. One of them is

\begin{footnotesize}
\begin{align*}
\left[ \begin{array}{ccccccc}
1 & 0 & 1 & 0 & 0 & 1 & 1 \\
0 & 0 & 1 & 1 & 1 & 1 & 0 \\
0 & 0 & 0 & 1 & 0 & 1 & 1
\end{array}
\right] .
\end{align*}
\end{footnotesize}
Exactly one of these 168 matrices is in reduced row echelon form.

\begin{footnotesize}
\begin{align*}
E(X)=\left[ \begin{array}{ccccccc}
1 & 0 & 0 & 0 & 1 & 1 & 0 \\
0 & 0 & 1 & 0 & 1 & 0 & 1 \\
0 & 0 & 0 & 1 & 0 & 1 & 1
\end{array}
\right] .
\end{align*}
\end{footnotesize}
\end{example}
\vspace{0.6cm}

A {\it Ferrers diagram} represents partitions as patterns of dots
with the $i$-th row having the same number of dots as the $i$-th
term in the partition~\cite{vLWi92,AnEr04,Sta86}. A Ferrers
diagram satisfies the following conditions.
\begin{itemize}
\item
The number of dots in a row is at most the number of dots in the
previous row.

\item All the dots are shifted to the right of the diagram.
\end{itemize}
The {\it number of rows (columns)} of the Ferrers diagram $\cF$ is
the number of dots in the rightmost column (top row) of $\cF$. If
the number of rows in the Ferrers diagram is $m$ and the number of
columns is $\eta$ we say that it is an $m \times \eta$ Ferrers
diagram.

If we read the Ferrers diagram by columns we get another partition
which is called the {\it conjugate} of the first one. If the
partition forms an $m \times \eta$ Ferrers diagram then the
conjugate partition form an $\eta \times m$ Ferrers diagram.

\vspace{0.2cm}
\begin{example}
\label{ex:Ferrers} Assume we have the partition $6+5+5+3+2$ of 21.
The $5 \times 6$ Ferrers diagram $\cF$ of this partition is given
by

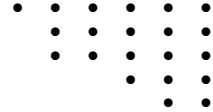
\begin{footnotesize}
\begin{align*}
\begin{array}{cccccc}
\bullet & \bullet & \bullet & \bullet & \bullet & \bullet \\
& \bullet & \bullet & \bullet & \bullet & \bullet \\
& \bullet & \bullet & \bullet & \bullet & \bullet \\
& & & \bullet & \bullet & \bullet \\
& & & & \bullet & \bullet
\end{array}
\end{align*}
\end{footnotesize}\\
The number of rows in $\cF$ is 5 and the number of columns is 6.
The conjugate partition is the partition $5+5+4+3+3+1$ of 21 and
its $6 \times 5$ Ferrers diagram is given by

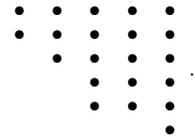
\begin{footnotesize}
\begin{align*}
\begin{array}{ccccc}
\bullet & \bullet & \bullet & \bullet &  \bullet \\
\bullet & \bullet & \bullet & \bullet & \bullet \\
& \bullet & \bullet & \bullet & \bullet \\
& & \bullet & \bullet & \bullet \\
& & \bullet & \bullet & \bullet  \\
& & & & \bullet
\end{array} .
\end{align*}
\end{footnotesize}
\end{example}
\vspace{0.3cm}

\begin{remark}
Our definition of Ferrers diagram is slightly different from the
usual definition~\cite{vLWi92,AnEr04,Sta86}, where the dots in
each row are shifted to the left of the diagram.
\end{remark}

Each $k$-dimensional subspace $X$ of $\F_q^n$ has an {\it
identifying vector} $v(X)$. $v(X)$ is a binary vector of length
$n$ and weight $k$, where the {\it ones} in $v(X)$ are in the
positions (columns) where $E(X)$ has the leading {\it ones} (of
the rows).

\vspace{0.2cm}
\begin{example}
Consider the 3-dimensional subspace $X$ of
Example~\ref{exm:running}. Its identifying vector is
$v(X)=1011000$.
\end{example}

\begin{remark}
We can consider an identifying vector $v(X)$ for some
$k$-dimensional subspace $X$ as a characteristic vector of a
$k$-subset. This coincides with the definition of rank- and
order-preserving map $\phi$ from $\Gr$ onto the lattice of subsets
of an $n$-set, given by Knuth~\cite{Knu71} and discussed by
Milne~\cite{Milne82}.
\end{remark}

The following lemma is easily observed.

\begin{lemma}
\label{lem:k_positions} Let $X$ be a $k$-dimensional linear
subspace of $\F_q^n$, $v(X)$ its identifying vector, and
$i_1,i_2,\ldots,i_k$ the positions in which $v(X)$ has {\it ones}.
Then for each nonzero element $u \in X$ the leftmost {\it one} in
$u$ is in position $i_j$ for some $1 \leq j \leq k$.
\end{lemma}
\begin{proof}
Clearly, for each $j$, $1 \leq j \leq k$, there exists an element
$u_j \in X$ whose leftmost {\it one} is in position $i_j$.
Moreover, $u_1,u_2,\ldots, u_k$ are linearly independent. Assume
the contrary, that there exists an element $u \in X$ whose
leftmost {\it one} is in position $\ell \notin \{ i_1 , \ldots ,
i_k \}$. This implies that $u,u_1,u_2,\ldots, u_k$ are linearly
independent and the dimension of $X$ is at least $k+1$, a
contradiction.
\end{proof}

The following result will play an important role in the proof that
our constructions for error-correcting codes in the projective
space have the desired minimum distance.
\begin{lemma}
\label{lem:dist_constant} If $X$ and $Y$ are two subspaces of
$\Ps$ with identifying vectors $v(X)$ and $v(Y)$, respectively,
then $d_S (X,Y) \geq d_H(v(X),v(Y))$, where $d_H(u,v)$ denotes the
Hamming distance between $u$ and $v$.
\end{lemma}
\begin{proof}
Let $i_1,...,i_r$ be the positions in which $v(X)$ has {\it ones}
and $v(Y)$ has {\it zeroes} and $j_1,...,j_s$ be the positions in
which $v(Y)$ has {\it ones} and $v(X)$ has {\it zeroes}. Clearly,
$r+s=d_H(v(X),v(Y))$. Therefore, by Lemma~\ref{lem:k_positions},
$X$ contains $r$ linearly independent vectors $u_1,...,u_r$ which
are not contained in $Y$. Similarly, $Y$ contains $s$ linearly
independent vectors which are not contained in $X$. Thus, $d_S
(X,Y) \geq r+s = d_H(v(X),v(Y))$.
\end{proof}

The {\it echelon Ferrers form} of a vector $v$ of length $n$ and
weight $k$, $EF(v)$, is the $k\times n$ matrix in reduced row
echelon form with leading entries (of rows) in the columns indexed
by the nonzero entries of $v$ and $"\bullet"$  in all entries
which do not have terminals {\it zeroes} or {\it ones}. A
$"\bullet"$ will be called in the sequel a {\it dot}. This
notation is also given in~\cite{vLWi92,Sta86}. The dots of this
matrix form the Ferrers diagram of $EF(v)$. If we substitute
elements of $\F_q$ in the dots of $EF(v)$ we obtain a
$k$-dimensional subspace $X$ of $\Ps$. $EF(v)$ will be called also
the echelon Ferrers form of $X$.

\vspace{0.2cm}

\begin{example}
For the vector $v=1001001$, the echelon Ferrers form $EF(v)$ is
the following $3 \times 7$ matrix,

\begin{footnotesize}
\begin{align*}
EF(v)=\left[ \begin{array}{ccccccc}
1 & \bullet & \bullet & 0 & \bullet & \bullet & 0 \\
0 & 0 & 0 & 1 & \bullet & \bullet & 0 \\
0 & 0 & 0 & 0 & 0 & 0 & 1
\end{array}
\right]~.
\end{align*}
\end{footnotesize}\\
$EF(v)$ has the following $2 \times 4$ Ferrers diagram
\begin{footnotesize}
\begin{align*}
\cF = \begin{array}{cccc}
\bullet & \bullet & \bullet & \bullet \\
&  & \bullet & \bullet
\end{array}~.
\end{align*}
\end{footnotesize}\\
\end{example}
\vspace{0.0cm}

Each binary word $v$ of length $n$ and weight $k$ corresponds to a
unique $k \times n$ matrix in an echelon Ferrers form. There are a
total of $\binom{n}{k}$ binary vectors of length $n$ and weight
$k$ and hence there are $\binom{n}{k}$ different $k \times n$
matrices in echelon Ferrers form.

\section{Ferrers Diagram Rank-Metric Codes}
\label{sec:partial}

In this section we start by defining the rank-metric codes. These
codes are strongly connected to constant dimension codes by a
lifting construction described by Silva, Kschischang, and
Koetter~\cite{SKK}. We define a new concept which is a Ferrers
diagram rank-metric code. Ferrers diagram rank-metric codes will
be the main building blocks of our projective space codes. These
codes present some questions which are of interest for themselves.

For two $m \times \eta$ matrices $A$ and $B$ over $\F_q$ the {\it
rank distance} is defined by
$$
d_R (A,B) \deff \text{rank}(A-B)~.
$$
A code $\cC$ is an $[m \times \eta,\varrho,\delta]$ rank-metric
code if its codewords are $m \times \eta$ matrices over $\F_q$,
they form a linear subspace of dimension $\varrho$ of $\F_q^{m
\times \eta}$, and for each two distinct codewords $A$ and $B$ we
have that $d_R (A,B) \geq \delta$. Rank-metric codes were well
studied~\cite{Gab85,Rot91,Del78}. It was proved (see~\cite{Rot91})
that for an $[m \times \eta,\varrho,\delta]$ rank-metric code
$\cC$ we have $\varrho \leq
\text{min}\{m(\eta-\delta+1),\eta(m-\delta+1)\}$. This bound is
attained for all possible parameters and the codes which attain it
are called {\it maximum rank distance} codes (or MRD codes in
short).

Let $v$ be a vector of length $n$ and weight $k$ and let $EF(v)$
be its echelon Ferrers form. Let $\cF$ be the Ferrers diagram of
$EF(v)$. $\cF$ is an $m \times \eta$ Ferrers diagram, $m \leq k$,
$\eta \leq n-k$. A code $\cC$ is an $[\cF,\varrho,\delta]$ {\it
Ferrers diagram rank-metric code} if all codewords are $m \times
\eta$ matrices in which all entries not in $\cF$ are {\it zeroes},
it forms a rank-metric code with dimension $\varrho$, and minimum
rank distance $\delta$. Let $\dim (\cF,\delta)$ be the largest
possible dimension of an $[\cF,\varrho,\delta]$ code.

\begin{theorem}
\label{thm:upper_rank} For a given $i$, $0 \leq i \leq \delta -1$,
if $\nu_i$ is the number of dots in $\cF$, which are not contained
in the first $i$ rows and are not contained in the rightmost
$\delta-1-i$ columns then $\text{min}_i \{ \nu_i \}$ is an upper
bound of $\dim (\cF,\delta)$.
\end{theorem}
\begin{proof}
For a given $i$, $0 \leq i \leq \delta-1$, let $\cA_i$ be the set
of the $\nu_i$ positions of $\cF$ which are not contained in the
first $i$ rows and are not contained in the rightmost $\delta-1-i$
columns. Assume the contrary that there exists an
$[\cF,\nu_i+1,\delta]$ code $\cC$. Let $\cB = \{ B_1,B_2,
\ldots,B_{\nu_i+1} \}$ be a set of $\nu_i+1$ linearly independent
codewords in $\cC$. Since the number of linearly independent
codewords is greater than the number of entries in $\cA_i$ there
exists a nontrivial linear combination $Y=\sum_{j=1}^{\nu_i+1}
\alpha_j B_j$ for which the $\nu_i$ entries of $\cA_i$ are equal
{\it zeroes}. $Y$ is not the all-zeroes codeword since the $B_i$'s
are linearly independent. $\cF$ has outside $\cA_i$ exactly $i$
rows and $\delta-i-1$ columns. These $i$ rows can contribute at
most $i$ to the rank of $Y$ and the $\delta-i-1$ columns can
contribute at most $\delta-i-1$ to the rank of $Y$. Therefore $Y$
is a nonzero codeword with rank less than $\delta$, a
contradiction.

Hence, an upper bound on $\dim (\cF,\delta )$ is $\nu_i$ for each
$0 \leq i \leq \delta-1$. Thus, an upper bound on the dimension
$\dim (\cF,\delta )$ is $\text{min}_i \{ \nu_i \}$.
\end{proof}

\begin{conjecture}
The upper bound of Theorem~\ref{thm:upper_rank} is attainable for
any given set of parameters $q$, $\cF$, and $\delta$.
\end{conjecture}

If we use $i=0$ or $i=\delta-1$ in Theorem~\ref{thm:upper_rank} we
obtain the following result.

\begin{cor}
\label{cor:upper_rank} An upper bound on $\dim (\cF,\delta )$ is
the minimum number of dots that can be removed from $\cF$ such
that the diagram remains with at most $\delta-1$ rows of dots or
at most $\delta-1$ columns of dots.
\end{cor}

\vspace{0.3cm}

\begin{remark}
$[m \times \eta,\varrho,\delta]$ MRD codes are one class of
Ferrers diagram rank-metric codes which attain the bound of
Corollary~\ref{cor:upper_rank} with equality. In this case the
Ferrers diagram has $m \cdot \eta$ dots.
\end{remark}

\begin{example}
Consider the following Ferrers diagram

\begin{footnotesize}
\begin{align*}
\cF= \begin{array}{cccc}
\bullet & \bullet & \bullet & \bullet \\
& & \bullet & \bullet\\
& &  & \bullet \\
& & & \bullet
\end{array}
\end{align*}
\end{footnotesize}\\
and $\delta=3$. By Corollary~\ref{cor:upper_rank} we have an upper
bound, $\dim (\cF,3) \leq 2$. But, if we use $i=1$ in
Theorem~\ref{thm:upper_rank} then we have a better upper bound,
$\dim (\cF,3) \leq 1$. This upper bound is attained with the
following generator matrix of an $[ \cF ,1,3]$ rank-metric code.
\[
\left(\begin{array}{cccc}
\bf 1 & \bf 0 & \bf 0 & \bf 0\\
0 &  0 & \bf 1 & \bf 0\\
0 & 0 & 0 & \bf 0\\
0 & 0 & 0 & \bf 1\end{array}\right).\]
\end{example}
\vspace{0.6cm}

When the bound of Theorem~\ref{thm:upper_rank} is attained? We
start with a construction of Ferrers diagram rank-metric codes
which attain the bound of Corollary~\ref{cor:upper_rank}. Assume
we have an $m \times \eta$, $m=\eta+\varepsilon$, $\varepsilon
\geq 0$, Ferrers diagram $\cF$ and that the minimum in the bound
of Corollary~\ref{cor:upper_rank} is obtained by removing all the
dots from the $\eta-\delta+1$ leftmost columns of $\cF$. Hence,
only the dots in the $\delta-1$ rightmost columns will remain. We
further assume that each of the $\delta-1$ rightmost columns of
$\cF$ have $m$ dots. The construction which follows is based on
the construction of MRD $q$-cyclic rank-metric codes given by
Gabidulin~\cite{Gab85}.

A code $\cC$ of length $m$ over $\F_{q^m}$ is called a {\it
$q$-cyclic code} if $(c_{0},c_{1},...,c_{m-1}) \in \cC$ implies
that $(c_{m-1}^{q},c_{0}^{q},...,c_{m-2}^{q})\in \cC$.

For a construction of $[m \times m, \varrho , \delta]$ rank-metric
codes, we use an isomorphism between the field with $q^m$
elements, $\F_{q^m}$, and the set of all $m$-tuples over $\F_q$,
$\F_q^m$. We use the obvious isomorphism by the representation of
an element $\alpha$ in the extension field $\F_{q^m}$ as $\alpha =
( \alpha_1 , \ldots, \alpha_m )$, where $\alpha_i$ is an element
in the ground field $\F_q$. Usually, we will leave to the reader
to realize when the isomorphism is used as this will be easily
verified from the context.

A codeword $c$ in an $[m \times m, \varrho , \delta]$ rank-metric
code $\cC$, can be represented by a vector $c=(c_0 , c_1 , \ldots
, c_{m-1})$, where $c_i \in \F_{q^m}$ and the generator matrix $G$
of $\cC$ is an $K \times m$ matrix, $\varrho = mK$. It was proved
by Gabidulin~\cite{Gab85} that if $\cC$ is an MRD $q$-cyclic code
then the generator polynomial of $\cC$ is the linearized
polynomial
$G(x)=\overset{m-K}{\underset{i=0}{\sum}}g_{i}x^{q^{i}}$, where
$g_i \in \F_{q^m}$, $g_{m-K}=1$, $m=K+\delta-1$, and its generator
matrix $G$ has the form

\begin{scriptsize}
\begin{align*}
\left(\begin{array}{cccccccc}
g_{0} & g_{1} & \cdots & g_{m-K-1} & 1 & 0 & \cdots & \cdots \\
0 & g_{0}^q & g_{1}^q & \cdots & g_{m-K-1}^q & 1 & \cdots & \cdots \\
0 & 0 & g_{0}^{q^2} & \cdots & \cdots & g_{m-K-1}^{q^2} & 1 & \cdots \\
\cdots & \cdots & \cdots & \cdots & \cdots & \cdots & \cdots  & \cdots \\
0 & \cdots & \cdots & \cdots & \cdots & \cdots &
g_{m-K-1}^{q^{K-1}} & 1
\end{array}\right).
\end{align*}
\end{scriptsize}

Hence, a codeword $c \in \cC$, $c \in (\F_{q^m})^m$, derived from
the information word $(a_0,a_1, \ldots,a_{K-1})$, where $a_i \in
\F_{q^m}$, i.e. $c = (a_0,a_1, \ldots,a_{K-1})G$, has the form

\begin{equation*}
c= (a_0 g_0 , a_0 g_1 + a_1 g_0^q , \ldots,a_{K-2}+a_{K-1}
g_{m-K-1}^{q^{K-1}} , a_{K-1})~.
\end{equation*}

We define an $[m \times \eta, m(\eta-\delta+1) , \delta]$
rank-metric code $\cC'$, $m = \eta + \varepsilon$, derived from
$\cC$ as follows:

$$
\cC' = \{ (c_0 , c_1 , \ldots , c_{\eta-1}) ~:~ (0,\ldots,0,c_0 ,
c_1 , \ldots , c_{\eta-1}) \in \cC \} .
$$

\begin{remark}
$\cC'$ is also an MRD code.
\end{remark}

We construct an $[\cF, \ell , \delta]$ Ferrers diagram rank-metric
code $\cC_{\cF} \subseteq \cC'$, where $\cF$ is an $m \times \eta$
Ferrers diagram. Let $\gamma_i$, $1 \leq i \leq \eta$, be the
number of dots in column $i$ of $\cF$, where the columns are
indexed from left to right. A codeword of $\cC_{\cF}$ is derived
from a codeword of $c \in \cC$ by satisfying a set of $m$
equations implied by

\begin{equation}
\label{eq:equations}
\begin{array}{c}\left(a_{0}g_{0},a_{0}g_{1}+a_{1}g_{0}^{q},\ldots,a_{K-2}+a_{K-1}g_{m-K-1}^{q^{K-1}},a_{K-1}\right)\\

=\left(\overset{\varepsilon}{\overbrace{\begin{array}{c}
0\\
\vdots\\
0\end{array}~\ldots~\begin{array}{c}
0\\
\vdots\\
0\end{array}}~}f_{1}~\ldots
f_{K-\varepsilon}\overset{\delta-1}{~\overbrace{\begin{array}{c}
\bullet\\
\vdots\\
\bullet\end{array}~\ldots~\begin{array}{c}
\bullet\\
\vdots\\
\bullet\end{array}}}\right) \end{array},
\end{equation}

\noindent where $f_i = ( \underset{\gamma_i} {\underbrace{\bullet
\cdots \bullet}} ~ \underset{m-\gamma_i}{\underbrace{0 \cdots
0}})^T$ is a column vector of length $m$, $1 \leq i \leq
K-\varepsilon$, and $u^T$ denotes the transpose of the vector $u$.
It is easy to verify that $\cC_{\cF}$ is a linear code.

By (\ref{eq:equations}) we have a system of $m=K+\delta-1$
equations with $K$ variables, $a_0,a_1, \ldots, a_{K-1}$. The
first $\varepsilon$ equations implies that $a_i=0$ for $0 \leq i
\leq \varepsilon -1$. The next $K-\varepsilon=\eta-\delta+1$
equations determine the values of the $a_i$'s, $\varepsilon \leq i
\leq K-1$, as follows. From the next equation $a_\varepsilon
g_0^{q^{\varepsilon}}=( \underset{\gamma_1}{\underbrace{\bullet
\cdots \bullet}}~ \underset{m-\gamma_1}{\underbrace{00...0}})^T$
(this is the next equation after we substitute $a_i=0$ for $0 \leq
i \leq \varepsilon-1$), we have that $a_\varepsilon$ has
$q^{\gamma_1}$ solutions in $\F_{q^m}$, where each element of
$\F_{q^m}$ is represented as an $m$-tuple over $\F_q$. Given a
solution of $a_\varepsilon$, the next equation $a_\varepsilon
g_0^{q^{\varepsilon}}+a_{\varepsilon+1} g_1^{q^{\varepsilon+1}}=
(\underset{\gamma_2} {\underbrace{\bullet \cdots
\bullet}}~\underset{m-\gamma_2}{\underbrace{00...0}})^T$ has
$q^{\gamma_2}$ solutions for $a_{\varepsilon+1}$. Therefore, we
have that $a_0,a_1, \ldots, a_{K-1}$ have
$q^{\sum_{i=1}^{K-\varepsilon} \gamma_i}$ solutions and hence the
dimension of $\cC_{\cF}$ is $\sum_{i=1}^{K-\varepsilon} \gamma_i$
over $\F_q$. Note, that since each of the $\delta-1$ rightmost
columns of $\cF$ have $m$ dots, i.e. $\gamma_i=m$,
$K-\varepsilon+1 \leq i \leq \eta$ (no {\it zeroes} in the related
equations) it follows that any set of values for the $a_i$'s
cannot cause any contradiction in the last $\delta-1$ equations.
Also, since the values of the $K$ variables $a_0,a_1, \ldots,
a_{K-1}$ are determined for the last $\delta-1$ equations, the
values for the related $(\delta-1)m$ dots are determined. Hence
they do not contribute to the number of solutions for the set of
$m$ equations. Thus, we have

\begin{theorem}
\label{thm:bound_attain}
Let $\cF$ be an $m \times \eta$, $m \geq \eta$, Ferrers diagram.
Assume that each one of the rightmost $\delta-1$ columns of $\cF$
has $m$ dots, and the $i$-th column from the left of $\cF$ has
$\gamma_i$ dots. Then $\cC_{\cF}$ is an $[ \cF , \sum_{i=1}^{\eta
- \delta +1} \gamma_i , \delta ]$ code which attains the bound of
Corollary~\ref{cor:upper_rank}.
\end{theorem}

\vspace{0.3cm}

\begin{remark}
For any solution of $a_0,a_1, \ldots, a_{K-1}$ we have that $
(a_0,a_1, \ldots,a_{K-1})G = (0,\ldots,0,c_0 , c_1 , \ldots ,
c_{\eta-1}) \in \cC$ and $(c_0 , c_1 , \ldots , c_{\eta-1}) \in
\cC_{\cF}$.
\end{remark}

\begin{remark}
For any $[m \times \eta, m(\eta-\delta+1) , \delta]$ rank-metric
code $\cC'$, the codewords which have {\it zeroes} in all the
entries which are not contained in $\cF$ form an $[ \cF ,
\sum_{i=1}^{\eta - \delta +1} \gamma_i , \delta ]$ code. Thus, we
can use also any MRD codes, e.g. the codes described
in~\cite{Rot91}, to obtain a proof for
Theorem~\ref{thm:bound_attain}.
\end{remark}

\begin{remark}
Since $\cC_{\cF}$ is a subcode of an MRD code then we can use the
decoding algorithm of the MRD code for the decoding of our code.
Also note, that if $\cF$ is an $m \times \eta$, $m < \eta$,
Ferrers diagram then we apply our construction for the $\eta
\times m$ Ferrers diagram of the conjugate partition.
\end{remark}

When $\delta =1$ the bounds and the construction are trivial. If
$\delta=2$ then by definition the rightmost column and the top row
of an $m \times \eta$ Ferrers diagram always has $m$ dots and
$\eta$ dots, respectively. It implies that the bound of
Theorem~\ref{thm:upper_rank} is always attained with the
construction if $\delta =2$. This is the most interesting case
since in this case the improvement of our constant dimension codes
compared to the codes in~\cite{KK,SKK} is the most impressive (see
subsection~\ref{sec:parameters}). If $\delta >2$ the improvement
is relatively small, but we will consider this case as it is of
interest also from a theoretical point of view. Some constructions
can be given based on the main construction and other basic
constructions. We will give two simple examples for $\delta=3$.

\vspace{0.2cm}
\begin{example}
\label{ex:ferrers1} Consider the following Ferrers diagram

\begin{footnotesize}
\begin{align*}
\cF= \begin{array}{cccc}
\bullet & \bullet & \bullet & \bullet \\
& \bullet & \bullet & \bullet\\
& & \bullet & \bullet \\
& & & \bullet
\end{array}
\end{align*}
\end{footnotesize}
The upper bound on $\dim (\cF,3)$ is 3. It is attained with the
following basis with three $4 \times 4$ matrices.
\[
\left(\begin{array}{cccc}
\bf 0 & \bf 1 & \bf 0 & \bf 0\\
0 & \bf 0 & \bf 1 & \bf 0\\
0 & 0 & \bf 0 & \bf 0\\
0 & 0 & 0 & \bf 1\end{array}\right),\left(\begin{array}{cccc}
\bf 0 & \bf 0 & \bf 0 & \bf 1\\
0 & \bf 1 & \bf 0 & \bf 0\\
0 & 0 & \bf 1 & \bf 0\\
0 & 0 & 0 & \bf 0\end{array}\right),\left(\begin{array}{cccc}
\bf 1 & \bf 0 & \bf 0 & \bf 0\\
0 & \bf 1 & \bf 0 & \bf 0\\
0 & 0 & \bf 0 & \bf 1\\
0 & 0 & 0 & \bf 1\end{array}\right).\]

\end{example}
\vspace{0.6cm}

\begin{example}
Consider the following Ferrers diagram

\begin{footnotesize}
\begin{align*}
\cF= \begin{array}{cccc}
\bullet & \bullet & \bullet & \bullet \\
& \bullet & \bullet & \bullet \\
& \bullet & \bullet & \bullet \\
& & & \bullet
\end{array}
\end{align*}
\end{footnotesize}
The upper bound on $\dim (\cF,3)$ is 4. It is attained with the
basis consisting of four $4 \times 4$ matrices, from which three
are from Example~\ref{ex:ferrers1} and the last one is
\[
\left(\begin{array}{cccc}
\bf 1 & \bf 0 & \bf 1 & \bf 0\\
0 & \bf 0 & \bf 0 & \bf 1\\
0 & \bf 1 & \bf 0 & \bf 1\\
0 & 0 & 0 & \bf 0\end{array}\right).\]
\end{example}
\vspace{0.6cm}

As for more constructions, some can be easily generated by the
interested reader, but whether the upper bound of
Theorem~\ref{thm:upper_rank} can be attained for all parameters
remains an open problem.

\section{Error-Correcting Constant Dimension Codes}
\label{sec:eccGrass}

In this section we will describe our multilevel construction. The
construction will be applied to obtain error-correcting constant
dimension codes, but it can be adapted to construct
error-correcting projective space codes without any modification.
This will be discussed in the next section. We will also consider
the parameters and decoding algorithms for our codes. Without loss
of generality we will assume that $k \leq n-k$. This assumption
can be made as a consequence of the following
lemma~\cite{EV,XiFu07}.

\begin{lemma}
If $\C$ is an $(n,M,2 \delta ,k)_q$ constant dimension code then
$\C^\perp = \{ X^\perp ~:~ X \in \C \}$, where $X^\perp$ is the
orthogonal subspace of $X$, is an $(n,M,2 \delta ,n-k)_q$ constant
dimension code.
\end{lemma}

\subsection{Lifted codes}

Koetter and Kschischang~\cite{KK} gave a construction for constant
dimension Reed-Solomon like codes. This construction can be
presented more clearly in terms of rank-metric codes~\cite{SKK}.
Given an $[k \times (n-k),\varrho,\delta]$ rank-metric code $\cC$
we form an $(n,q^\varrho,2\delta,k)_q$ constant dimension code
$\C$ by {\it lifting} $\cC$, i.e., $\C = \{ [I_k ~ A] ~:~ A \in
\cC \}$, where $I_k$ is the $k \times k$ identity
matrix~\cite{SKK}. We will call the code $\C$ the {\it lifted
code} of $\cC$. Usually $\C$ is not maximal and it can be
extended. This extension requires to design rank-metric codes,
where the shape of a codeword is a Ferrers diagram rather than an
$k \times (n-k)$ matrix. We would like to use the largest possible
Ferrers diagram rank-metric codes. In the appropriate cases, e.g.
when $\delta=2$, we will use the codes constructed in
Section~\ref{sec:partial} for this purpose.

Assume we are given an echelon Ferrers form $EF(v)$ of a binary
vector $v$, of length $n$ and weight $k$, with a Ferrers diagram
$\cF$ and a Ferrers diagram rank-metric code $\cC_{\cF}$.
$\cC_{\cF}$ is lifted to a constant dimension code $\C_v$ by
substituting each codeword $A \in \cC_{\cF}$ in the columns of
$EF(v)$ which correspond to the {\it zeroes} of $v$. Note, that
depending on $\cF$ it might implies conjugating $\cF$ first.
Unless $v$ starts with an {\it one} and ends with a {\it zero}
(the cases in which $\cF$ is a $k \times (n-k)$ Ferrers diagram)
we also need to expand the matrices of the Ferrers diagram
rank-metric code to $k \times (n-k)$ matrices (which will be
lifted), where $\cF$ is in their upper right corner (and the new
entries are {\it zeroes}). As an immediate consequence
from~\cite{SKK} we have.

\begin{lemma}
\label{lem:dist_lift} If $\cC_{\cF}$ is an $[ \cF , \varrho ,
\delta ]$ Ferrers diagram rank-metric code then its lifted code
$\C_v$, related to an $k \times n$ echelon Ferrers form $EF(v)$,
is an $(n, q^\varrho , 2\delta , k)_q$ constant dimension code.
\end{lemma}

\vspace{0.2cm}

\begin{example}
For the word $v= 1110000 $, its echelon Ferrers form
\begin{footnotesize}
\begin{align*}
EF(v)=\left[ \begin{array}{ccccccc}
1 & 0 & 0 & \bullet & \bullet & \bullet & \bullet \\
0 & 1 & 0 & \bullet & \bullet & \bullet & \bullet \\
0 & 0 & 1 & \bullet & \bullet & \bullet & \bullet
\end{array}
\right]~,
\end{align*}
\end{footnotesize}\\
the $3 \times 4$ matrix
\[
\left(\begin{array}{cccc}
\bf 1 & \bf 0 & \bf 1 & \bf 0\\
\bf 0 & \bf 0 & \bf 0 & \bf 1\\
\bf 0 & \bf 0 & \bf0 & \bf 0
\end{array}\right)\]
is lifted to the 3-dimensional subspace with the $3 \times 7$ generator matrix
\begin{footnotesize}
\begin{align*}
\left[ \begin{array}{ccccccc}
1 & 0 & 0 & \bf 1 & \bf 0 & \bf 1 & \bf 0\\
0 & 1 & 0 & \bf 0 & \bf 0 & \bf 0 & \bf 1\\
0 & 0 & 1 & \bf 0 & \bf 0 & \bf0 & \bf 0
\end{array}
\right]~.
\end{align*}
\end{footnotesize}\\
For the word $v= 1001001 $, its echelon Ferrers form
\begin{footnotesize}
\begin{align*}
EF(v)=\left[ \begin{array}{ccccccc}
1 & \bullet & \bullet & 0 & \bullet & \bullet & 0 \\
0 & 0 & 0 & 1 & \bullet & \bullet & 0 \\
0 & 0 & 0 & 0 & 0 & 0 & 1
\end{array}
\right]~,
\end{align*}
\end{footnotesize}\\
the $2 \times 4$ matrix
\[
\left(\begin{array}{cccc}
\bf 1 & \bf 0 & \bf 1 & \bf 0\\
0 &  0 & \bf 0 & \bf 1
\end{array}\right)\]
is lifted to the 3-dimensional subspace with the $3 \times 7$ generator matrix
\begin{footnotesize}
\begin{align*}
\left[ \begin{array}{ccccccc}
1 & \bf 1 & \bf 0 & 0 & \bf 1 & \bf 0 & 0\\
0 & 0 & 0 & 1 & \bf 0 & \bf 1 & 0\\
0 & 0 & 0 & 0 & 0 & 0 & 1
\end{array}
\right]~.
\end{align*}
\end{footnotesize}\\

\end{example}
\vspace{0.6cm}

The code described in~\cite{SKK} is the same as the code described
in Section~\ref{sec:partial}, where the identifying vector is $(1
\cdots 1 0 \cdots 0)$. If our lifted codes are the codes described
in Section~\ref{sec:partial} then the same decoding algorithm can
be applied. Therefore, the decoding in~\cite{SKK} for the
corresponding constant dimension code can be applied directly to
each of our lifted constant dimension codes in this case, e.g. it
can always be applied when $\delta=2$. It would be worthwhile to
permute the coordinates in a way that the identity matrix $I_k$
will appear in the first $k$ columns, from the left, of the reduced row echelon
form, and $\cF$ will appear in the upper right corner of the $k
\times n$ matrix. The reason is that the decoding of~\cite{SKK} is
described on such matrices.

\subsection{Multilevel construction}

Assume we want to construct an $(n,M,2\delta,k)_q$ constant
dimension code $\C$.

The first step in the construction is to choose a binary constant
weight code ${\bf C}$ of length $n$, weight $k$, and minimum
distance $2 \delta$. This code will be called the {\it skeleton
code}. Any constant weight code can be chosen for this purpose,
but different skeleton codes will result in different constant
dimension codes with usually different sizes. The best choice for
the skeleton code ${\bf C}$ will be discussed in the next
subsection. The next three steps are performed for each codeword
$c \in {\bf C}$.

The second step is to construct the echelon Ferrers form $EF(c)$.

The third step is to construct an $[\cF,\varrho,\delta]$ Ferrers
diagram rank-metric code $\cC_{\cF}$ for the Ferrers diagram $\cF$
of $EF(c)$. If possible we will construct a code as described in
Section~\ref{sec:partial}.

The fourth step is to lift $\cC_{\cF}$ to a constant dimension
code $\C_c$, for which the echelon Ferrers form of $X \in \C_c$ is
$EF(c)$.

Finally,
$$
\C = \bigcup_{c \in {\bf C}} \C_c ~.
$$

As an immediate consequence of Lemmas~\ref{lem:dist_constant}
and~\ref{lem:dist_lift} we have the following theorem.
\begin{theorem}
$\C$ is an $(n,M,2\delta,k)_q$ constant dimension code,
where $M= \sum_{c \in {\bf C}} |\C_c| $.
\end{theorem}

\vspace{0.2cm}
\begin{example}
\label{ex:n-6_k=3} Let $n=6$, $k=3$, and ${\bf C}= \{ 111000,~
100110,~ 010101,~ 001011 \}$ a constant weight code of length 6,
weight 3, and minimum Hamming distance 4. The echelon Ferrers
forms of these 4 codewords are

\begin{footnotesize}
\begin{align*}
EF(111000)=\left[ \begin{array}{cccccc}
1 & 0 & 0 & \bullet & \bullet & \bullet \\
0 & 1 & 0 & \bullet & \bullet & \bullet \\
0 & 0 & 1 & \bullet & \bullet & \bullet
\end{array}
\right]
\end{align*}
\end{footnotesize}

\begin{footnotesize}
\begin{align*}
EF(100110)=\left[ \begin{array}{cccccc}
1 & \bullet & \bullet & 0 & 0 & \bullet \\
0 & 0 & 0 & 1 & 0 & \bullet \\
0 & 0 & 0 & 0 & 1 & \bullet
\end{array}
\right]
\end{align*}
\end{footnotesize}

\begin{footnotesize}
\begin{align*}
EF(010101)=\left[ \begin{array}{cccccc}
0 & 1 & \bullet & 0 & \bullet & 0 \\
0 & 0 & 0 & 1 & \bullet & 0 \\
0 & 0 & 0 & 0 & 0 & 1
\end{array}
\right]
\end{align*}
\end{footnotesize}

\begin{footnotesize}
\begin{align*}
EF(001011)=\left[ \begin{array}{cccccc}
0 & 0 & 1 & \bullet & 0 & 0 \\
0 & 0 & 0 & 0 & 1 & 0 \\
0 & 0 & 0 & 0 & 0 & 1
\end{array}
\right] ~.
\end{align*}
\end{footnotesize}

By Theorem~\ref{thm:bound_attain}, the Ferrers diagrams of these four
echelon Ferrers forms yield Ferrers diagram rank-metric codes of
sizes 64, 4, 2, and 1, respectively. Hence, we obtain a
$(6,71,4,3)_2$ constant dimension code $\C$.

\end{example}

\vspace{0.6cm}

\begin{remark} A $(6,74,4,3)_2$ code was obtained by computer
search~\cite{KoKu08}. Similarly, we obtain a $(7,289,4,3)_2$ code.
A $(7,304,4,3)_2$ code was obtained by computer
search~\cite{KoKu08}.
\end{remark}

\vspace{0.2cm}
\begin{example}
\label{ex:c844} Let ${\bf C}$ be the codewords of weight 4 in the
[8,4,4] extended Hamming code with the following parity-check
matrix.

\begin{footnotesize}
\begin{align*}
\left[ \begin{array}{cccccccc}
0 & 0 & 0 & 0 & 1 & 1 & 1 & 1 \\
0 & 0 & 1 & 1 & 0 & 0 & 1 & 1 \\
0 & 1 & 0 & 1 & 0 & 1 & 0 & 1 \\
1 & 1 & 1 & 1 & 1 & 1 & 1 & 1
\end{array}
\right]
\end{align*}
\end{footnotesize}
${\bf C}$ has 14 codewords with weight 4. Each one of these
codewords is considered as an identifying vector for the echelon
Ferrers forms from which we construct the final $(8,4573,4,4)_2$
code $\C$. The fourteen codewords of ${\bf C}$ and their
contribution for the final code $\C$ are given in the following
table. The codewords are taken in lexicographic order.

\medskip
\begin{tabular}{|c|>{\centering}p{3cm}|>{\centering}p{3cm}|}
\hline
 & $c \in \cC$ & size of $\C_c$ \tabularnewline
\hline \hline 1 & 11110000 & 4096\tabularnewline \hline 2 &
11001100 & 256\tabularnewline \hline 3 & 11000011 &
16\tabularnewline \hline 4 & 10101010 & 64\tabularnewline \hline 5
& 10100101 & 16\tabularnewline \hline 6 & 10011001 &
16\tabularnewline \hline 7 & 10010110 & 16\tabularnewline \hline 8
& 01101001 & 32\tabularnewline \hline 9 & 01100110 &
16\tabularnewline \hline 10 & 01011010 & 16\tabularnewline \hline
11 & 01010101 & 8\tabularnewline \hline 12 & 00111100 &
16\tabularnewline \hline 13 & 00110011 & 4\tabularnewline \hline
14 & 00001111 & 1\tabularnewline \hline
\end{tabular}
\end{example}
\vspace{0.6cm}

\subsection{Code parameters}
\label{sec:parameters}

We now want to discuss the size of our constant dimension code,
the required choice for the skeleton code ${\bf C}$, and compare the size
of our codes with the size of the codes constructed
in~\cite{KK,SKK}.

The size of the final constant dimension code $\C$ depends on the
choice of the skeleton code ${\bf C}$. The identifying vector with
the largest size of corresponding rank-metric code is
$\underset{k}{\underbrace{1 \cdots 1}}\underset{n-k}{\underbrace{0
\cdots 0}}$. The corresponding $[k \times (n-k),\ell,\delta]$
rank-metric code has dimension $\ell=(n-k)(k-\delta+1)$ and hence
it contributes $q^{(n-k)(k-\delta+1)}$ $k$-dimensional subspaces
to our final code $\C$. These subspaces form the codes
in~\cite{KK,SKK}. The next identifying vector which contributes
the most number of subspaces to $\C$ is
$\underset{k-\delta}{\underbrace{11...1}}\underset{\delta}
{\underbrace{0 \cdots 0}}\underset{\delta}{\underbrace{11...1}}
\underset{n-k-\delta}{\underbrace{000...00}}$. The number of
subspaces it contributes depends on the bounds presented in
Section~\ref{sec:partial}. The rest of the code $\C$ usually has
less codewords from those contributed by these two. Therefore, the
improvement in the size of the code compared to the code
of~\cite{KK} is not dramatic. But, for most parameters our codes
are larger than the best known codes. In some cases, e.g. when
$\delta=k$ our codes are as good as the best known codes
(see~\cite{EV}) and suggest an alternative construction.  When
$k=3$, $\delta=4$, and reasonably small $n$, the cyclic codes
constructed in~\cite{EV,KoKu08} are larger.

Two possible alternatives for the best choice for the skeleton
code ${\bf C}$ might be of special interest. The first one is for
$k=4$ and $n$ which is a power of two. We conjecture that the best
skeleton code is constructed from the codewords with weight 4 of
the extended Hamming code for which the columns of the
parity-check matrix are given in lexicographic order. We
generalize this choice of codewords from the Hamming code by
choosing a constant weight lexicode~\cite{CoSl86}. Such a code is
constructed as follows. All vectors of length $n$ and weight $k$
are listed in lexicographic order. The code ${\bf C}$ is generated
by adding to the code ${\bf C}$ one codeword at a time. At each
stage, the first codeword of the list that does not violate the
distance constraint with the other codewords of ${\bf C}$, is
joined to ${\bf C}$. Lexicodes are not necessarily the best
constant weight codes. For example, the largest constant weight
code of length 10 and weight 4 is 30, while the lexicode with the
same parameters has size 18. But, the constant dimension code
derived from the lexicode is larger than any constant dimension
code derived from any related code of size 30.

The following table summarized the sizes of some of our codes
compared to previous known codes. In all these codes we have
started with a constant weight lexicode in the first step of the
construction.

\vspace{0.2cm}
\begin{tabular}{|c|c|c|c|c|c|}
\hline $q$ & $d_S (\C)$ & $n$ & $k$ & code size\cite{KK} & size of
our code\tabularnewline \hline \hline 2 & 4 & 9 & 4 & $2^{15}$ &
$2^{15}$+4177\tabularnewline \hline 2 & 4 & 10 & 5 & $2^{20}$ &
$2^{20}$+118751\tabularnewline \hline 2 & 4 & 12 & 4 & $2^{24}$ &
$2^{24}$+2290845\tabularnewline \hline 2 & 6 & 10 & 5 & $2^{15}$ &
$2^{15}$+73\tabularnewline \hline
% 2 & 6 & 12 & 6 & $2^{24}$ & $2^{24}$+34120*\tabularnewline \hline
2 & 6 & 13 & 4 & $2^{18}$ & $2^{18}$+4357\tabularnewline \hline 2
& 8 & 21 & 5 & $2^{32}$ & $2^{32}$+16844809 \tabularnewline \hline
3 & 4 & 7 & 3 & $3^{8}$ & $3^{8}$+124\tabularnewline \hline 3 & 4
& 8 & 4 & $3^{12}$ & $3^{12}$+8137\tabularnewline \hline 4 & 4 & 7
& 3 & $4^{8}$ & $4^{8}$+345\tabularnewline \hline 4 & 4 & 8 & 4 &
$4^{12}$ & $4^{12}$+72529\tabularnewline \hline
\end{tabular}
\vspace{0.1cm}

%\noindent *some of the Ferrers diagram rank-metric codes which
%were taken, do not meet the bound of Theorem~\ref{thm:upper_rank}.

\subsection{Decoding}

The decoding of our codes is quite straightforward and it mainly
consists of known decoding algorithms. As we used a multilevel
coding we will also need a multilevel decoding. In the first step
we will use a decoding for our skeleton code and in the
second step we will use a decoding for the rank-metric codes.

Assume the received word was a $k$-dimensional subspace $Y$. We
start by generating its reduced row echelon form $E(Y)$. Given
$E(Y)$ it is straightforward to find the identifying vector
$v(Y)$. Now, we use the decoding algorithm for the constant weight
code to find the identifying vector $v(X)$ of the submitted
$k$-dimensional subspace $X$. If no more than $\delta -1$ errors
occurred then we will find the correct identifying vector. This
claim is an immediate consequence of
Lemma~\ref{lem:dist_constant}.

In the second step of the decoding we are given the received
subspace $Y$, its identifying vector $v(Y)$, and the identifying
vector $v(X)$ of the submitted subspace $X$. We consider the
echelon Ferrers form $EF(v(X))$, its Ferrers diagram $\cF$, and
the $[\cF, \varrho, \delta ]$ Ferrers diagram rank-metric code
associated with it. We can permute the columns of $EF(v(X))$, and
use the same permutation on $Y$, in a way that the identity matrix
$I_k$ will be in the left side. Now, we can use the decoding of
the specific rank-metric code. If our rank-metric codes are those
constructed in Section~\ref{sec:partial} then we can use the
decoding as described in~\cite{SKK}. It is clear now that the
efficiency of our decoding depends on the efficiency of the
decoding of our skeleton code and the efficiency of the decoding
of our rank-metric codes. If the rank-metric codes are MRD codes
then they can be decoded efficiently~\cite{Gab85,Rot91}. The same is
true if the Ferrers diagram metric codes are those constructed in
Section~\ref{sec:partial} as they are subcodes of MRD codes and the
decoding algorithm of the related MRD code applied to them too.

There are some alternative ways for our decoding, some of which
improve on the complexity of the decoding. For example we can make
use of the fact that most of the code is derived from two
identifying vectors or that most of the rank-metric codes are of
relatively small size. One such case can be when all the identity
matrices of the echelon Ferrers forms are in consecutive columns
of the codeword (see~\cite{Ska08}). We will not discuss it as the
related codes hardly improve on the codes in~\cite{KK,SKK}.

Finally, if we allow to receive a word which is an
$\ell$-dimensional subspace $Y$, $k-\delta+1 \leq \ell \leq
k+\delta-1$, then the same procedure will work as long as $d_S
(X,Y) \leq \delta -1$. This is a consequence of the fact that the
decoding algorithm of~\cite{SKK} does not restrict the dimension
of the received word.

\section{Error-Correcting Projective Space Codes}
\label{sec:eccProj}

In this section our goal will be to construct large codes in $\Ps$
which are not constant dimension codes. We first note that the
multilevel coding described in Section~\ref{sec:eccGrass} can be
used to obtain a code in $\Ps$. The only difference is that we
should start in the first step with a general binary code of
length $n$ in the Hamming space as a skeleton code. The first question which will
arise in this context is whether the method is as good as for
constructing codes in $\Gr$. The answer can be inferred from the
following example.

\vspace{0.2cm}
\begin{example}
\label{ex:c394} Let $n=7$ and $d=3$, and consider the [7,4,3]
Hamming code with the parity-check matrix

\begin{footnotesize}
\begin{align*}
\left[ \begin{array}{ccccccc}
0 & 0 & 0 & 1 & 1 & 1 & 1 \\
0 & 1 & 1 & 0 & 1 & 1 & 0 \\
1 & 0 & 1 & 1 & 0 & 1 & 0
\end{array}
\right] .
\end{align*}
\end{footnotesize}

By using the multilevel coding with this Hamming code we obtain a
code with minimum distance 3 and size 394 in
$\smash{{\sP\kern-2.0pt}_2\kern-0.5pt(7)}$.
\end{example}
\vspace{0.6cm}

As we shall see in the sequel this code is much smaller than a
code that will be obtained by puncturing. We have also generated
codes in the projective space based on the multilevel
construction, where the skeleton code is a lexicode. The
constructed codes appear to be much smaller than the codes
obtained by puncturing. Puncturing of a code $\C$ (or union of
codes with different dimensions and the required minimum distance)
obtained in Section~\ref{sec:eccGrass} results in a projective
space code $\C'$. If the minimum distance of $\C$ is $2 \delta$
then the minimum distance of $\C'$ is $2 \delta -1$. $\C'$ has a
similar structure to a code obtained by the multilevel
construction (similar structure in the sense that the identifying
vectors of the codewords can form a skeleton code). But the
artificial "skeleton code" can be partitioned into pair of
codewords with Hamming distance one, while the distance between
two codewords from different pairs is at least $2 \delta -1$. This
property yields larger codes by puncturing, sometimes with double
size, compared to codes obtained by the multilevel construction.

\subsection{Punctured codes}
\label{sec:punct}

Puncturing and punctured codes are well known in the Hamming
space. An $(n,M,d)$ code in the Hamming space is a code of length
$n$, minimum Hamming distance $d$, and $M$ codewords. Let ${\bf
C}$ be an $(n,M,d)$ code in the Hamming space. Its punctured code
${\bf C}'$ is obtained by deleting one coordinate of ${\bf C}$.
Hence, there are $n$ punctured codes and each one is an
$(n-1,M,d-1)$ code. In the projective space there is a very large
number of punctured codes for a given code $\C$ and in contrary to
the Hamming space the sizes of these codes are usually different.

Let $X$ be an $\ell$-subspace of $\F_q^n$ such that the unity
vector with an {\it one} in the $i$-th coordinate is not an element
in $X$. The {\it $i$-coordinate puncturing} of $X$, $\Delta_i(X)$,
is defined as the $\ell$-dimensional subspace of $\F_q^{n-1}$
obtained from $X$ by deleting coordinate $i$ from each vector in
$X$. This puncturing of a subspace is akin to puncturing a code
${\bf C}$ in the Hamming space by the $i$-th coordinate.

Let $\C$ be a code in $\Ps$ and let $Q$ be an $(n-1)$-dimensional
subspace of $\F_q^n$. Let $E(Q)$ be the $(n-1) \times n$ generator
matrix of $Q$ (in reduced row echelon form) and let $\tau$ be the
position of the unique {\it zero} in $v(Q)$. Let $v \in \F_q^n$ be an
element such that $v \notin Q$. We define the {\it punctured} code

$$\C'_{Q,v} = \C_Q \cup \C_{Q,v}~,$$
where
$$\C_Q=\left\{ \Delta_\tau (X) ~:~ X \in \C, ~ X \subseteq Q\right\})$$
and
$$\C_{Q,v}= \left\{ \Delta_\tau (X \cap Q)~:\: X \in \C,~v\in X \right\}~.$$

\begin{remark} If $\C$ was constructed by the multilevel construction of
Section~\ref{sec:eccGrass} then the codewords of $\C_Q$ and
$\C_{Q,v}$ can be partitioned into related lifted codes of Ferrers
diagram rank-metric codes. Some of these codes are cosets of the
linear Ferrers diagram rank-metric codes.
\end{remark}

The following theorem can be easily verified.

\begin{theorem}
The punctured code $\C'_{Q,v}$ of an $(n,M,d)_q$ code $\C$ is an
$(n-1,M',d-1)_q$ code.
\end{theorem}

\begin{remark} The code $\tilde{\C} = \left\{ X ~:~ X \in \C, ~ X \subseteq
Q\right\}) \cup \left\{ X \cap Q~:\: X \in \C,~v\in X \right\}$ is
an $(n,M',d-1)_q$ code whose codewords are contained in $Q$. Since
$Q$ is an $(n-1)$-dimensional subspace it follows that there is an
isomorphism $\varphi$ such that $\varphi(Q) = \F_q^{n-1}$. The
code $\varphi (\tilde{\C}) = \{ \varphi (X) ~:~ X \in \tilde{\C}
\}$ is an $(n-1,M',d-1)_q$ code. The code $\C'_{Q,v}$ was obtained
from $\tilde{\C}$ by such isomorphism which uses the
$\tau$-coordinate puncturing on all the vectors of $Q$.
\end{remark}

\begin{example}
\label{ex:c573} Let $\C$ be the $(8,4573,4,4)_2$ code given in
Example~\ref{ex:c844}. Let $Q$ be the 7-dimensional subspace whose
$7 \times 8$ generator matrix is

\[
\left(\begin{array}{cccccc}
1 & 0 & \ldots & 0 & 0\\
0 & 1 & \ldots & 0 & 0\\
\vdots & \vdots & \ddots & \vdots & \vdots\\
0 & 0 & \ldots & 1 & 0\end{array}\right).\]

By using puncturing with $Q$ and $v=1000001$ we obtain a code
$\C'_{Q,v}$ with minimum distance 3 and size 573.  By adding to
$\C'_{Q,v}$ two codewords, the null space $\{ 0 \}$ and $\F_2^7$
we obtained a $(7,575,3)_2$ code in $\mathcal{P}_{2}(7)$. The
following tables show the number of codewords which were obtained
from each of the identifying vectors with weight 4 of
Example~\ref{ex:c844}.

\vspace{0.2cm}
\begin{tabular}{|>{\centering}p{1.1in}|>{\centering}p{1.1in}|}
\hline \multicolumn{1}{|c}{$\mathbb{C}_{Q}$} &
\multicolumn{1}{c|}{}\tabularnewline \hline \hline identifying
vector & addition to $\mathbb{C}_{Q}$ \tabularnewline \hline
11110000 & 256\tabularnewline \hline 11001100 & 16\tabularnewline
\hline 10101010 & 8\tabularnewline \hline 10010110 &
2\tabularnewline \hline 01100110 & 4\tabularnewline \hline
01011010 & 2\tabularnewline \hline 00111100 & 1\tabularnewline
\hline
\end{tabular}

\smallskip{}

\begin{tabular}{|>{\centering}p{1.1in}|>{\centering}p{1.1in}|}
\hline \multicolumn{1}{|c}{$\mathbb{C}_{Q,v}$ } &
\multicolumn{1}{c|}{$v=10000001$}\tabularnewline \hline \hline
identifying vector & \multicolumn{1}{p{1.1in}|}{addition to
$\mathbb{C}_{Q,v}$}\tabularnewline \hline 11110000 &
256\tabularnewline \hline 11001100 & 16\tabularnewline \hline
11000011 & 1\tabularnewline \hline 10101010 & 4\tabularnewline
\hline 10100101 & 2\tabularnewline \hline 10011001 &
4\tabularnewline \hline 10010110 & 1\tabularnewline \hline
\end{tabular}
\vspace{0.3cm}

\noindent The Ferrers diagram rank-metric codes in some of the
entries must be chosen (if we want to construct the same code by
a multilevel construction) in a clever way and not directly as given
in Section~\ref{sec:partial}. We omit their description from lack
of space and leave it to the interested reader.

\end{example}
\vspace{0.6cm}

The large difference between the sizes of the codes of
Examples~\ref{ex:c394} and~\ref{ex:c573} shows the strength of
puncturing when applied on codes in $\Ps$.

\subsection{Code parameters}

First we ask, what is the number of codes which can be defined in
this way from $\C$? $Q$ is an $(n-1)$-dimensional subspace of
$\F_q^n$ and hence it can be chosen in $\frac{q^n-1}{q-1}$
different ways. There are $\frac{q^n-q^{n-1}}{q-1}=q^{n-1}$
distinct way to choose $v \notin Q$ after $Q$ was chosen. Thus, we
have that usually puncturing of a code $\C$ in $\Ps$ will result
in $\frac{q^{2n-1}-q^{n-1}}{q-1}$ different punctured codes.

\begin{theorem}
If $\C$ is an $(n,M,d,k)_q$ code then there exists an
$(n-1,M',d-1)_q$ code $\C'_{Q,v}$ such that $M' \geq
\frac{M(q^{n-k}+q^k-2)}{q^n-1}$.
\end{theorem}
\begin{proof}
As before, $Q$ can be chosen in $\frac{q^n-1}{q-1}$ different
ways. By using basic enumeration, it is easy to verify that each
$k$-dimensional subspace of $\Ps$ is contained in
$\frac{q^{n-k}-1}{q-1}$~~$(n-1)$-dimensional subspaces of $\Ps$.
Thus, by a simple averaging argument we have that there exists an
$(n-1)$-dimensional subspace $Q$ such that $|\C_Q|\geq
M\frac{q^{n-k}-1}{q^n-1}$.

There are $M- |\C_Q|$ codewords in $\C$ which are not contained in
$Q$. For each such codeword $X \in \C$ we have $\dim (X \cap Q)=k-1$.
Therefore, $X$ contains $q^k-q^{k-1}$ vectors which do not belong
to $Q$. In $\F_q^n$ there are $q^n-q^{n-1}$ vectors which do not
belong to $Q$. Thus, again by using simple averaging argument we
have that there exist an $(n-1)$-dimensional subspace $Q \subset
\F_q^n$ and $v\notin Q$ such that $|\C_{Q,v}|\geq
\frac{(M-|\C_Q|)(q^k-q^{k-1})}{q^n-q^{n-1}}=\frac{M-|\C_Q|}{q^{n-k}}$.

Therefore, there exists an $(n-1,M',d-1)_q$ code $\C'_{Q,v}$ such
that $M'=|\C_Q|+|\C_{Q,v}|\geq
\frac{|\C_Q|q^{n-k}+M-|\C_Q|}{q^{n-k}}=
\frac{(q^{n-k}-1)|\C_Q|+M}{q^{n-k}} \geq
\frac{(q^{n-k}-1)M(q^{n-k}-1)+M(q^n-1)}{(q^n-1)q^{n-k}}=
\frac{M(q^{n-k}+q^k-2)}{q^n-1}$.
\end{proof}
\vspace{0.2cm}

Clearly, choosing the $(n-1)$-dimensional subspace $Q$ and the
element $v$ in a way that $\C'_{Q,v}$ will be maximized is
important in this context. Example~\ref{ex:c573} can be
generalized in a very simple way. We start with a
$(4k,q^{2k(k+1)},2k,2k)_q$ code obtained from the codeword
$\underset{2k}{\underbrace{1 \cdots 1}}\underset{2k}{\underbrace{0
\cdots 0}}$ in the multilevel approach. We apply puncturing with
the $(n-1)$-dimensional subspace $Q$ whose $(n-1) \times n$
generator matrix is

\[
\left(\begin{array}{cccccc}
1 & 0 & \ldots & 0 & 0\\
0 & 1 & \ldots & 0 & 0\\
\vdots & \vdots & \ddots & \vdots & \vdots\\
0 & 0 & \ldots & 1 & 0\end{array}\right).\]

It is not difficult to show that in the $[(2k) \times (2k) ,
2k(k+1), k]$ rank-metric code $\cC$ there are $q^{2k^2}$ codewords
with {\it zeroes} in the last column and $q^{2k^2}$ codewords with
{\it zeroes} in the first row. There is also a codeword whose
first row ends with a {\it one}. If $u$ is this first row which
ends with a {\it one} there are $q^{2k^2}$ codewords whose first
row is $u$. We choose $v$ to be $v=1\underset{2k-1}{\underbrace{0
\cdots 0}}u$. By using puncturing with $Q$ and $v$ we have
$|\C_Q|=q^{2k^2}$ and $|\C_{Q,v}|=q^{2k^2}$. Hence, $\C'_{Q,v}$ is
a $(4k-1,2q^{2k^2},2k-1)_q$ code in $\mathcal{P}_{q}(4k-1)$. By
using more codewords from the constant weight code in the
multilevel approach and adding the null space and $\F_q^{4k-1}$ to
the code we construct a slightly larger code with the same
parameters.

\subsection{Decoding}

We assume that $\C$ is an $(n,M,d)_q$ code and that all the
dimensions of the subspaces in $\C$ have the same parity which
implies that $d=2 \delta$. This assumption makes sense as these
are the interesting codes on which puncturing is applied,
similarly to puncturing in the Hamming space. We further assume
for simplicity that w.l.o.g. if $E(Q)$ is the $(n-1) \times n$
generator matrix of $Q$ then the first $n-1$ columns are linearly
independent, i.e., $E(Q)= [ I ~ u ]$, where $I$ is an $(n-1)
\times (n-1)$ unity matrix and $u$ is a column vector of length
$n-1$.

Assume that the received word from a codeword $X'$ of $\C'_{Q,v}$
is an $\ell$-dimensional subspace $Y'$ of $\F_q^{n-1}$. The first
step will be to find a subspace $Z$ of $\F_q^n$ on which we can apply
the decoding algorithm of $\C$. The result of this decoding will reduced
to the $(n-1)$-dimensional subspace $Q$ and punctured
to obtain the codeword of $\C'_{Q,v}$. We start
by generating from $Y'$ an $\ell$-dimensional subspace $Y \subset
Q$ of $\F_q^n$. This is done by appending a symbol to the end of
each vector in $Y'$ by using the generator matrix $E(Q)$ of $Q$.
If a generator matrix $E(Y')$ is given we can do this process only
to the rows of $E(Y')$ to obtain the generator matrix $E(Y)$ of
$Y$. We leave to the reader the verification that the generator
matrix of $Y$ is formed in its reduced row echelon form.

\begin{remark} If the {\it zero} of $v(Q)$ is in coordinate $\tau$ then
instead of appending a symbol to the end of the codeword we insert
a symbol at position $\tau$.
\end{remark}

Let $\ell$ be the dimension of $Y'$ and assume $p$ is the parity
of the dimension of any subspace in $\C$, where $p=0$ or $p=1$.
Once we have $Y$ we distinguish between two cases to form a new
subspace $Z$ of $\F_q^n$.

\noindent {\bf Case 1:} $\delta$ is even.
\begin{itemize}
\item If $\ell \equiv p~(mod~2)$ then $Z= Y \cup (v + Y)$.

\item If $\ell \not\equiv p~(mod~2)$ then $Z=Y$.
\end{itemize}

\noindent {\bf Case 2:} $\delta$ is odd.
\begin{itemize}
\item If $\ell \equiv p~(mod~2)$ then $Z= Y$.

\item If $\ell \not\equiv p~(mod~2)$ then $Z= Y \cup (v + Y)$.
\end{itemize}

Now we use the decoding algorithm of the code $\C$ with the word
$Z$. The algorithm will produce as an output a codeword $\cX$. Let
$\tilde{X} = \cX \cap Q$ and $\tilde{X}'$ be the subspace of
$\F_q^{n-1}$ obtained from $\tilde{X}$ by deleting the last entry
of $\tilde{X}$. We output $\tilde{X}'$ as the submitted codeword
$X'$ of $\C'_{Q,v}$. The correctness of the decoding algorithm is
an immediate consequence from the following theorem.

\begin{theorem}
If $d_S (X',Y') \leq \delta-1$ then $\tilde{X}'=X'$.
\end{theorem}
\begin{proof}
Assume that $d_S (X',Y') \leq \delta-1$. Let $X \subseteq Q$ be
the word obtained from $X'$ by appending a symbol to the end of
each vector in $X'$ (this can be done by using the generator
matrix $E(Q)$ of $Q$). If $u \in X' \cap Y'$ then we append the same
symbol to $u$ to obtain the element of $X$ and to obtain the element
of $Y$. Hence, $d_S (X,Y)=d_S (X',Y') \leq \delta
-1$. If $d_S (X,Y) \leq \delta -2$ then $d_S (X,Z) \leq d_S (X,Y)
+1 \leq \delta -1$. Now, note that if $\delta-1$ is odd then $Z$
does not have the same parity as the dimensions of the subspaces
in $\C$ and if $\delta-1$ is even then $Z$ has the same parity as
the dimensions of the subspaces in $\C$. Therefore, if $d_S
(X,Y)=\delta-1$ then by the definition of $Z$ we have $Z=Y$ and
hence $d_S (X,Z)=\delta-1$. Therefore, the decoding algorithm of
$\C$ will produce as an output the unique codeword $\cX$ such that
$d_S ( \cX , Z) \leq \delta -1$, i.e., $X = \cX$. $X'$ is obtained
by deleting the last entry is each vector of $X \cap Q$;
$\tilde{X}'$  is obtained by deleting the last entry is each
vector of $\cX \cap Q$. Therefore, $\tilde{X}'=X'$.
\end{proof}

The constant dimension codes constructed in
Section~\ref{sec:eccGrass} have the same dimension for all
codewords. Hence, if $\C$ was constructed by our multilevel
construction then its decoding algorithm can be applied on the
punctured code $\C'_{Q,v}$.

\section{Conclusion and Open Problems}
\label{sec:conclude}

A multilevel coding approach to construct codes in the
Grassmannian and the projective space was presented. The method
makes usage of four tools, an appropriate constant weight code,
the reduced row echelon form of a linear subspace, the Ferrers
diagram related to this echelon form, and rank-metric codes
related to the Ferrers diagram. Some of these tools seem to be
important and interesting for themselves in general and
particularly in the connection of coding in the projective space.
The constructed codes by our method are usually the best known
today for most parameters. We have also defined the puncturing
operation on codes in the projective space. We applied this
operation to obtain punctured codes from our constant dimension
codes. These punctured codes are considerably larger than codes
constructed by any other method. The motivation for considering
these codes came from network coding~\cite{ACLY,HMKKESL,HKMKE} and
error-correction in network coding~\cite{KK,CaYe06,YeCa06}. It
worth to mention that the actual dimensions of the
error-correcting codes needed for network coding are much larger
than the dimensions given in our examples. Clearly, our method
works also on much higher dimensions as needed for the real
application.

The research on coding in the projective space is only in its
first steps and many open problems and directions for further
research are given in our references. We focus now only in
problems which are directly related to our current research.

\begin{itemize}
\item Is there a specification for the best constant weight code
which should be taken for our multilevel approach? Our discussion
on the Hamming code and lexicodes is a first step in this
direction.

\item Is the upper bound of Theorem~\ref{thm:upper_rank} attained
for all parameters? Our constructions for optimal Ferrers diagram
rank-metric codes suggest that the answer is positive.

\item How far are the codes constructed by our method from
optimality? The upper bounds on the sizes of codes in the
Grassmannian and the projective space are still relatively much
larger than the sizes of our codes~\cite{EV,KK,XiFu07}. The
construction of cyclic codes in~\cite{EV} suggests that indeed
there are codes which are relatively much larger than our codes.
But, we believe that in general the known upper bounds on the
sizes of codes in the projective space are usually much larger
than the actual size of the largest codes. Indeed, solution for
this question will imply new construction methods for
error-correcting codes in the projective space.

\end{itemize}

\section*{Acknowledgment}

The authors wish to thank Ronny Roth and Alexander Vardy for many
helpful discussions. They also thank the anonymous reviewers for their valuable comments.

%\bibliography{allbib,extra}

\end{document}